\documentclass[reqno]{amsart}
\usepackage{graphicx} 
\usepackage{amsaddr}

\usepackage[dvipsnames]{xcolor}

\usepackage{float}
\usepackage[shortlabels]{enumitem}
\usepackage[numbers,sort&compress]{natbib}
\usepackage{braket}
\usepackage{amsthm}
\usepackage{mathtools}
\usepackage{url}
\usepackage{physics}
\usepackage{graphicx}
\usepackage[left=24mm,right=24mm,top=35mm]{geometry} 
\usepackage[T1]{fontenc}
\usepackage{bm}
\usepackage{tabularx}
\newcommand*{\eh}{\mathrm{End\ }\mathcal{H}}
\newcommand*{\Ad}{\mathrm{Ad}}

\def\ad{^{\dagger}}

\newcommand{\fsnull}[1]{}
\newcommand{\old}[1]{}

\usepackage[makeroom]{cancel}
\usepackage[toc,page]{appendix}
\definecolor{C1}{RGB}{52, 89, 149}
\definecolor{C2}{RGB}{251, 77, 61}
\definecolor{C3}{RGB}{3, 206, 164}
\definecolor{C4}{RGB}{202, 21, 81}
\definecolor{C5}{RGB}{180,77,155}
\usepackage{hyperref}
\hypersetup{colorlinks=true, linkcolor=C5, citecolor=C5, urlcolor=C5}

\usepackage{tikz}
\tikzset{every picture/.style=remember picture}

\usepackage{pgfplots}

\usepackage[utf8]{inputenc}
\usepackage{graphicx}
\usepackage{xcolor}
\usepackage{amsmath}
\usepackage{amsthm}
\usepackage{bm}
\usepackage{bbm}
\usepackage{comment}
\usepackage{appendix}
\usepackage{mathdots}
\usepackage{lipsum}
\usepackage{verbatim}
\usepackage{nccmath}
\usepackage{amsfonts}
\usepackage{thm-restate}
\usepackage{thmtools}
\usepackage{capt-of}
\usepackage{booktabs}

\newcounter{graph}
\renewcommand{\thegraph}{\Alph{graph}}

\usepackage{amssymb}
\usepackage{dsfont}

\renewcommand{\geq}{\geqslant}
\renewcommand{\leq}{\leqslant}

\newcommand{\ot}{\otimes}
\newcommand{\ts}{^{\otimes 2}}

\newcommand{\bs}{\textsf{BS}}

\newcommand{\lm}{\lambda }

\newcommand{\sg}{\sigma }

\newcommand{\om}{\omega }
\newcommand{\Om}{\Omega }

\tikzset{qubit/.style={circle,draw,thick,fill=blue!5,minimum size=7mm,inner sep=1pt},
  leaf/.style={font=\small,inner sep=2pt},
  spectator/.style={draw=orange!75!black,text=orange!65!black},
  edge/.style={thick},port/.style={font=\scriptsize,fill=white,inner sep=1pt}}

\usetikzlibrary{arrows.meta}
\DeclareMathOperator*{\expect}{\mathbb{E}}

\newcommand{\mcl}{\mathcal{L}}

\newcommand{\mco}{\mathcal{O}}

\newcommand{\mch}{\mathcal{H}}

\newcommand{\mce}{\mathcal{E}}

\newcommand{\mbso}{\mathbb{SO}}

\newcommand{\mbu}{\mathbb{U}}

\newcommand{\mbc}{\mathbb{C}}

\newcommand{\mst}{\mathsf{T}}

\def\be{\begin{equation}}
\def\ee{\end{equation}}
\def\bs{\begin{split}}
\def\e{\end{split}}
\def\ba{\begin{eqnarray}}
\def\bea{\begin{eqnarray}}

\def\tea{\end{eqnarray}}
\def\ea{\end{eqnarray}}
\def\eea{\end{eqnarray}}

\def\b{\beta}

\def\b{\beta}

\def\tn{^\otimes n}
\def\tk{^\otimes k}

\def\b{\beta}

\def\tn{^{\otimes n}}

\def\tk{^{\otimes k}}

\newcommand{\id}{\mathds{1}}

\renewcommand{\b}{\beta}

\newcommand{\sdket}[1]{| #1 \rangle\!\rangle}

\def\mg{\mathsf{MG}(n)}

\def\sg{\sigma}

\def\rtg{{\rm rt}(\mathsf{G})}

\def\be{\begin{equation}}
\def\te{\end{equation}}
\def\ee{\end{equation}}
\def\ba{\begin{eqnarray}}
\def\bea{\begin{eqnarray}}

\def\tea{\end{eqnarray}}
\def\ea{\end{eqnarray}}
\def\eea{\end{eqnarray}}

\pgfplotsset{compat=1.18}

\begin{document}

\bibliographystyle{IEEEtran}

\title[Strong matchgate designs in nearly optimal depth]{Strong matchgate designs in nearly optimal depth}

\author{\vspace{-5mm} Maxwell  West\,\textsuperscript{1,2},  M. Cerezo\,\textsuperscript{2,3} \MakeLowercase{and}  Mart\'{i}n Larocca\,\textsuperscript{1,2}}
\address{\textsuperscript{1}Theoretical Division, Los Alamos National Laboratory, Los Alamos, New Mexico 87545, USA}
\address{\textsuperscript{2}Quantum Science Center, Oak Ridge, TN 37931, USA}
\address{\textsuperscript{3}Information Sciences, Los Alamos National Laboratory, Los Alamos, New Mexico 87545, USA}

\begin{abstract}{Understanding the resources required to generate approximately random unitaries over various groups is a natural goal of quantum information theory. With respect to one notion of approximation, that of a \textit{design}, it is known that the full unitary group can be approximated in logarithmic depth by one-dimensional circuits of nearest-neighbour 2-local gates. On the other hand, remarkably, circuits with this  connectivity cannot form designs over the matchgate group in sublinear depth. Here we show that this dramatic slowdown can disappear when using a general qubit connectivity graph $\mathsf{G}$ of \textit{routing number ${\rm rt}(\mathsf{G})$}. Indeed, in this setting one  can obtain (strong) $\varepsilon$-approximate relative error matchgate $k$-designs in depth $\mco(k^2{\rm rt}(\mathsf{G})\log n\log(n/\varepsilon))$. For all-to-all connectivity, ${\rm rt}(\mathsf{G})=2$.  Our construction is conceptually simple, involving a random walk on the matchgate group, and no ancillae. As a technical byproduct, we improve upon the state of the art for \textit{fermionic routing}, obtaining an   $\mco({\rm rt}(\mathsf{G})\log n)$ depth router. Additionally, for $k=3$, we obtain an \textit{exact} strong matchgate design in  $\mco({\rm rt}(\mathsf{G})\log n)$ depth, again without ancillae. Under all-to-all connectivity, our fermionic router and 3-designs are optimal. Notably, our  results imply that  quantum algorithms  for fermionic tomography which require drawing from a matchgate 3-design may be exponentially sped up on quantum computers with all-to-all connectivity, relative to their strictly one-dimensional counterparts.
}
\end{abstract}

\maketitle

\tableofcontents

\section{Introduction}
Unitary designs, ensembles of unitaries which reproduce moments of the unitary Haar distribution, are fundamental to  quantum computation and information~\cite{schuster2024random,laracuente2024approximate,cui2025unitary,brandao2016local,haferkamp2022random,ma2024how,west2025no,grevink2025will,west2026ambient,dankert2009exact,harrow2009random,chen2024incompressibility,gross2007evenly,haah2025short,parella2026strong,west2024random,deneris2024exact}. They provide a framework within which to study the scrambling capacity of quantum systems~\cite{roberts2017chaos,cotler2017chaos,schuster2025strong,dowling2023scrambling,brown2012scrambling},   constitute important primitives in quantum algorithms~\cite{huang2020predicting,elben2022randomized,knill2008randomized,combes2017logical,helsen2022matchgate,magesan2011scalable}, and   even inform our understanding of phenomena in  quantum gravity~\cite{hayden2007black,sekino2008fast,lashkari2013towards}. Counterintuitively, it has been discovered in recent years that ensembles of one-dimensional circuits of nearest-neighbour 2-local gates with depth scaling just as $\mco(k\,{\rm polylog}\, (k) \log(n/\varepsilon)) $, can form $\varepsilon$-approximate unitary $k$-designs~\cite{schuster2024random,laracuente2024approximate}. On an   all-to-all connected lattice, this scaling can be improved to doubly logarithmic in $n$~\cite{cui2025unitary,du2026low}. \\

The notion of unitary designs  invites natural generalisation in several directions.  One such generalisation is to that of designs over other groups; in general, we shall say that an ensemble constitutes an (approximate) $k$-design over a compact group $G$ if its first $k$ moments coincide (approximately) with those of the Haar measure $\mu_G$ on $G$. There has recently been some interest in the complexity of constructing such group designs, where it has been found~\cite{west2025no,grevink2025will,west2026ambient} that a surprising difficulty often emerges. Indeed, and in contrast to the logarithmic depth at   which approximate designs can  form   in the   unitary case,   one dimensional sublinear-depth local nearest-neighbour ensembles cannot form approximate 2-designs over the matchgate, orthogonal, or   symplectic groups; neither can they form approximate 4-designs over the Clifford group. Strikingly, even allowing these circuits to constitute unitaries from beyond the group over which they are attempting to form a design does not enable one to overcome this sublinear depth barrier~\cite{west2026ambient}. The heart of the difficulty may be traced to the possessing by  these groups of invariant states in low-degree tensor powers of their defining representations; essentially these states may be used to introduce  sublinearity-detecting lightcones. Indeed, suppose that there exists a state $\ket \Psi\in\mch\ts$ such that,  for all unitaries $U$ in a group $G$, one has $U\ts\ket\Psi=\ket\Psi$. Then, for any $V\in \mbu(\mch)$, one further has that $U\ts (V\ot\id_\mch)\ket\Psi =   ((UVU\ad)\ot\id_\mch) \ket\Psi $. If one is working with nearest-neighbour one-dimensional  circuits,   the Heisenberg-evolved $UVU\ad$ then reveals   lightcone effects which lead to $\Om(n)$ lower bounds on the depths of approximate designs~\cite{west2025no,grevink2025will,west2026ambient}. \\

Operationally, an ensemble $\mce$ forming a $k$-design over a group $G$ may be formalised as the statement that, given access to $U\tk$ for a $U$ which is promised to be drawn either as $U\sim\mce$ or $U\sim\mu_G$,  it is impossible to tell which is true (with a low but   non-zero probability of success then being captured by the notion of an \textit{approximate} design). Various conditions on the experiments one is permitted to do on $U\tk$  then lead to different notions of a design. For example, if one is permitted only to perform a POVM directly on $U\tk$, then indistinguishability from the Haar case certifies $\mce$ as an \textit{additive error design}. A stronger criterion, which for example allows one to perform adaptive queries in which results of experiments on some copies of $U$ inform the experiments performed on subsequent ones, is captured by the notion of a \textit{measurable error design}; requiring further that the probabilities of experiments   be approximated multiplicatively instead of additively gives the even stronger notion of a \textit{relative error design}~\cite{cui2025unitary,mele2023introduction}, on which we will focus on this work. Independently of which metric one is using to quantify approximation, a natural strengthening of the notion of a design, to that of a so-called \textit{strong design},  is to consider experiments which can probe not only elements $U\sim\mce$, but also $U^\mst,\overline{U},$ and $U\ad$~\cite{schuster2025strong,parella2026strong}. Under this access model, the lightcone   obstructions  which applied already to the matchgate, orthogonal,   symplectic, and Clifford groups under the standard model apply in the full unitary case as well, immediately ruling out the possibility of sublinear depth strong unitary designs on a one-dimensional nearest neighbour model.  \\

The fact that lightcone arguments  rule out rapid design formation over the matchgate  group, and  rapid strong design formation over the unitary group, suggest that one   might in both cases consider the ability of circuits with all-to-all connectivity, for which lightcones become vacuous beyond logarithmic depth, to form such group designs. Indeed, it has been shown that under such connectivity strong unitary designs can form in  logarithmic depth~\cite{schuster2025strong,parella2026strong}; the lightcones of the 1D case are in a sense the ``only obstruction'' to sublinear depth formation. It is natural, then, to wonder if the same is true in the matchgate case; in this work, we show that the answer is yes. More generally, on an arbitrary qubit connectivity graph $G$ with \textit{routing number} $\rtg$~\cite{yuan2025full}, we construct  $\varepsilon$-approximate relative error strong matchgate $k$-designs in depth $\mco(k^2\rtg \log n\log(n/\varepsilon))$, and exact strong matchgate $3$-designs in  $\mco({\rm rt}(\mathsf{G})\log n)$ depth. In particular, the routing number (to be defined formally in Section~\ref{sec:prelim}) of an all-to-all connected graph is equal to two, so that under all-to-all qubit connectivity our 3-designs achieve the optimal $\Theta(\log n)$ depth. This implies that, for applications which require (approximate) matchgate designs, and for which the generation of the design elements is the   most expensive part of the procedure, quantum computers with all-to-all connectivity can enjoy an exponential speed up over their one-dimensional counterparts. 

\section{Preliminaries}\label{sec:prelim}
In this section we briefly review the properties of group designs, and the matchgate group, that will be important in this work.  
For an ensemble $\mu$ of unitaries and nonnegative integers $p,q$, we define the \textit{mixed moment channel}
\begin{equation}\label{eq:mom}
 \Phi_\mu^{(p,q)}(X)=\expect_{U\sim\mu}\left[(U^{\otimes p}\otimes\overline{U}{}^{\otimes q})X(U^{\otimes p}\otimes\overline{U}{}^{\otimes q})^\dagger
 \right].
\end{equation}
For a (compact) group $G\subseteq\mbu(d)$, we write $\mu_G$ for the  Haar measure on $G$.  The ordinary $k$\textsuperscript{th} moment operator~\cite{mele2023introduction} with respect to $G$ is thus the case $(p,q)=(k,0)$, and $\mu=\mu_G$.
For linear maps, $\Psi_1$ and $\Psi_2$, we denote $\Psi_2-\Psi_1$ being completely positive by $\Psi_1\preceq\Psi_2$. We say that an ensemble is a strong relative error $\varepsilon$-approximate $k$-design over $G$  if, for every $p+q=k$,
\begin{equation}
 (1-\varepsilon)\Phi_{\mu_G}^{(p,q)}\preceq\Phi_\mu^{(p,q)}\preceq(1+\varepsilon)\Phi_{\mu_G}^{(p,q)}.
 \label{eq:relative}
\end{equation}
Now, let $\mch\cong\mbc^{d}$ with $d=2^n$ be the Hilbert space of $n$ qubits, and let $V=\mathbb C^{2n}$ be the vector representation of $\mbso({2n})$.
The $4^{n}$  Majorana monomials $c_{a_1}\cdots c_{a_s}$, for $0\leq s\leq {2n}$ and $a_1<\cdots<a_s$,
form an orthogonal basis of $\eh$.  
We will  make the standard choice
\begin{equation}
 c_{2j-1}=Z_1\cdots Z_{j-1}X_j,\qquad c_{2j}=Z_1\cdots Z_{j-1}Y_j, \qquad 1\leq j\leq n.
\end{equation}
of Jordan-Wigner Majorana operators, which obey the canonical anticommutation relations $\{c_\mu,c_\nu\}=2\delta_{\mu\nu}\id$.
Next, for $a<b$ we define the 2-Majorana rotation
\begin{equation}
 R_{ab}(\vartheta)=\exp\left(-\frac{\vartheta}{2}c_ac_b\right),
 \label{eq:rotation}
\end{equation}
which acts on $c_a$ and $c_b$ as 
\begin{align}
 R_{ab}(\vartheta)c_aR_{ab}(\vartheta)^\dagger &=\cos\vartheta\,c_a+\sin\vartheta\,c_b,\nonumber\\
 R_{ab}(\vartheta)c_bR_{ab}(\vartheta)^\dagger&=-\sin\vartheta\,c_a+\cos\vartheta\,c_b,
 \label{eq:rot}
\end{align}
and fixes every other Majorana.  We define the matchgate group $\mathsf{MG}(n)\cong \operatorname{Spin}(2n)$ to be the group generated by these unitaries. As is well known, the adjoint action of a matchgate unitary on a  Majorana maps it to its image under the associated orthogonal matrix (i.e., the Majoranas furnish the vector representation $V$ of $\mbso(2n)$);  on a  product of $s$ Majoranas, it  acts as the $s$\textsuperscript{th} exterior power of that matrix.  This gives an  
  isomorphism
\begin{equation}\label{eq:iso}
 \eh\cong \bigoplus_{s=0}^{2n}\bigwedge^s V=:\bigwedge V.
\end{equation}
We will be   interested in the eigenvalues of the operator representing $R_{ab}(\vartheta)$ in various representations; it will however be convenient not to speak of eigenvalues \textit{per se}, but rather of \textit{charges}, where the charge corresponding to an eigenvalue $e^{-iq\vartheta}$ is $q$.
On the space  $V\cong \mbc^{2n}$ of Majoranas, a rotation in this $(a,b)$ plane then has   eigenvectors with charges $+1$ and $-1$ (one of each, the complex linear combinations of the Majoranas $a$ and $b$   which are fixed up to phase by the rotation (see Eq.~\eqref{eq:rot})), and  $2n-2$ other eigenvectors with charge zero (that is, eigenvalue $e^0 = 1$). Now, by the isomorphism of Eq.~\eqref{eq:iso}, a basis of ${\rm End}\,\mch$ may be obtained by choosing subsets of
these eigenvectors and taking their wedge product.  Each such wedge product is also  eigenvector of the Majorana rotation, with a charge given by the sum of the charges of its factors. Since each of the two eigenvectors with charges $+1$ and $-1$ can occur at most once, this resulting charge must be in
  $\{-1,0,1\}$. More generally, $R_{ab}(\vartheta)$ acts via the $k$-fold tensor product
representation on $({\rm End}\,\mch)\tk$; as charges are additive, they can in that representation take
every integer value from $-k$ to $k$.
Equivalently, defining the Hermitian generator $X_{ab}$ by $\Ad_{R_{ab}(\vartheta)\tk}=e^{-i\vartheta X_{ab}}$, a  vector has charge $q$ exactly  when it is an eigenvector of $X_{ab}$ with eigenvalue $q$. The bounding of the magnitude of these eigenvalues will be important later, when we consider certain spectral properties of our design constructions.\\

Next, we notice that in the case of ensembles whose elements belong to the matchgate group, the distinction between regular and strong matchgate designs in fact disappears. The reason for this is the possession by the matchgate group of a preserved bilinear form $\Om=X\ot Y\ot X\ot Y\ot X\ldots$~\cite{west2025no,west2026ambient}, which one can readily verify to satisfy $U\Om =\Om \overline{U}$, for any matchgate unitary $U$. But then, with $W_{p,q} = \id_{\mch}^{\ot p}\ot \Om^{\ot q}$, we have (recalling Eq.~\eqref{eq:mom})  
\begin{equation*}
\Phi_\mu^{(p,q)}=\Ad_{W_{p,q}}\circ\Phi_\mu^{(k,0)}
       \circ\Ad_{W_{p,q}};
\end{equation*}
as pre- and post-composition with unitary channels
preserve CP order, we have from Eq.~\eqref{eq:relative} that, at least for ensembles that consist of elements which are themselves within the matchgate group, being a strong matchgate $k$-design is equivalent to being a regular matchgate $k$-design. Operationally, the preserved bilinear form allows one to transform copies of $U$ into copies of $\overline{U}$, which can then be used in any experiment to distinguish a putative design from the Haar measure; copies of $U^\mst$ and $U\ad$ can be introduced by post-selecting on contractions with Bell states~\cite{schuster2025strong}. \\

Finally, we make precise the notion of the routing number $\rtg$ of a graph $\mathsf{G}=(V,E)$~\cite{yuan2025full}. Choose an ordering of the vertices, and ``attach'' the number $i$ to the $i$\textsuperscript{th} vertex (for all $i$). Suppose we are given an arbitrary ``target permutation'' $\pi\in S_{|V|}$, and that at each timestep, we can choose a matching $M\subseteq E$ of the edges and swap the numbers at the endpoints of these edges. There is some minimum number of timesteps needed in order for us to have rearranged the contents of the vertices into the order specified by $\pi$; the maximum of that number over all permutations is defined to be $\rtg$.  From the readily verified fact that any permutation may be written as the composition of two involutions, it for example follows that $\rtg=2$ if $G$ is all-to-all connected.

\section{Results}
As advertised, our main result is an explicit construction of strong
$\varepsilon$-approximate relative-error $n$-qubit matchgate $k$-designs,
with depth polylogarithmic in $n$ for fixed $k$, $\varepsilon$, and $\rtg$.
Our construction is a random walk on the matchgate group, with each step in the walk taking the form of a 2-Majorana rotation $R_{a,b}(\vartheta)$ (see Eq.~\eqref{eq:rotation}) for randomly chosen Majoranas $c_a,c_b,$ and a random rotation angle $\vartheta$. 
At a high level, we have three steps (see also Figure~\ref{fig:1}):
\begin{enumerate}[label=(\arabic*),leftmargin=*]
    \item \textit{Bounding walk length.} Randomly sample $m\in\mco(k^2 n\log (n/\varepsilon))$ pairs $1\leq a_i< b_i\leq 2n$, and $m$ angles $\vartheta_i$.
    \item \textit{Bounding the number of layers.}
    Check if  $\{R_{a_i,b_i}(\vartheta_i)\}_{i=1}^m$ can be parallelised into $ \mco(k^2 \log (n/\varepsilon))$ layers; if not, go back to step one. Return the concatenation of two such independently accepted walks. 
    \item \textit{Bounding the depth of each layer.} Compile each parallel layer of Majorana rotations into two qubit gates of depth $\mco(\rtg\log n)$ via a fermionic router.
\end{enumerate}

\vspace{4mm}

In a little more detail: In Step 1, the fact that it suffices to  take $m\in\mco(k^2 n\log (n/\varepsilon))$ follows from some Fourier  analysis of the random walk on the matchgate group induced by making these random 2-Majorana rotations. In Step 2, we need to control the  probability of a randomly chosen set of Majorana rotations failing to compress into  $ \mco(k^2 \log (n/\varepsilon))$ layers of commuting rotations. Essentially, we will see that the walks of Step 1  already give a strong  approximate design; the difficulty is that as they involve $\mco(k^2 n\log (n/\varepsilon))$ random gates, for some walks we will get unlucky and wind  up with a circuit that cannot be compressed below linear depth. The technical content of Step 2 is then to show that conditioning our sampling of random walks on their being compressible also gives a design (we will see that this is not quite true; instead of simply sampling from the set of walks conditioned on being compressible, we will need to independently sample two such walks, and return their concatenation). By the end of Step 2 we therefore have an ensemble of shallow, ``local'', circuits which constitute a design; the problem is that they are local only in the fermionic sense (i.e., consisting of 2-Majorana rotations). In Step 3 we show that, at a cost of a factor of $\rtg\log n$, these   circuits can be compiled into physical 2-local gates. \\

\begin{figure}
\setlength{\captionindent}{5pt}
    \centering
    \includegraphics[width=\linewidth]{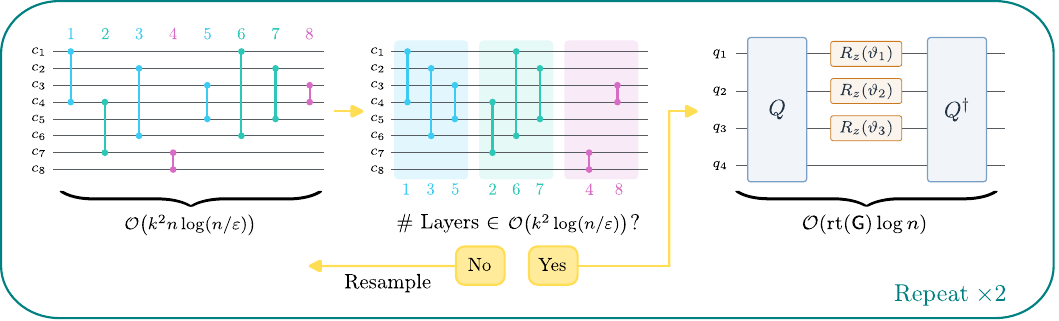}
    \caption{ Our construction proceeds in three steps: (1) Randomly sample $m\in\mco(k^2 n\log (n/\varepsilon))$ pairs $1\leq a_i< b_i\leq 2n$, and $m$ angles $\vartheta_i$; apply the corresponding Majorana rotations $\{R_{a_i,b_i}(\vartheta_i)\}_{i=1}^m$.  (2)
    Check if those rotations can be parallelised into $ \mco(k^2 \log (n/\varepsilon))$ layers of commuting gates; if not, go back to step one.  (3) 
    Compile each parallel layer of Majorana rotations into two qubit gates of depth $\mco(\rtg\log n)$ via a fermionic router. Concatenate   two such independently accepted circuits; the total depth is then given by  (\# Layers) $\times\, \mco(\rtg\log n) = \mco(k^2\rtg\log(n/\varepsilon)\log n)$. }
    \label{fig:1}
\end{figure}

Before formalising this argument properly, we note that   the main technical problem which arises in Step 3 is that of compiling an arbitrary fermionic permutation into 2-local physical gates, which is a task that has received some attention in the literature. For example, previous constructions have shown   that (assuming all-to-all connectivity) this can be done in depth  $\mco(\log^2 n)$~\cite{constantinides2025low}, and more recently, depth $\mco(\log^{1+a} n)$ for any $a>0$~\cite{aigner2026fermion}. When allowing the use of ancilla qubits, mid-circuit measurements, and classical feed-forward, it has been shown that the optimal depth of $\Theta(\log n)$ can be achieved~\cite{maskara2025fast}. It is then perhaps of some independent interest that we show

\begin{restatable}{thm}{thmpermcomp}\label{thm:perm_comp}
An arbitrary signed permutation of the $2n$ Majoranas can be   compiled into   depth   $\mco(\rtg\log n)$, without ancillae,   using only Cliffords. For all-to-all connectivity, this is optimal, and can be done efficiently.
\end{restatable}

\noindent
Appendix~\ref{sec:perm} is devoted to the proof of Theorem~\ref{thm:perm_comp}. As an immediate corollary, we have: 

\begin{restatable}{crl}{crlthreedes}\label{crl:threedes}
There is an exact  $n$-qubit matchgate $3$-design which can be compiled into a 2-local Clifford circuit depth $\mco(\rtg\log n)$. For all-to-all connectivity this is optimal, and the compilation is efficient.
\end{restatable}
\begin{proof}
The technical input we need here is that the uniform distribution on the intersection of the Clifford and matchgate groups (that is, the signed Majorana-permutations of determinant one) are known to form an exact matchgate 3-design~\cite{gargiulo2026pauli,wan2022matchgate,heyraud2024unified}.  To uniformly sample such a signed permutation (i.e., a $Q$ satisfying $Qc_iQ\ad =\sigma_i c_{\pi(i)}$ for all $1\leq i \leq 2n$, and signs $\sigma_i\in\{\pm 1\}$), we sample a uniform permutation $\pi\in \mathfrak{S}_{2n}$, choose $\sigma_1,\ldots,\sigma_{2n-1}$ randomly, and then set $\sigma_{2n}={\rm sgn}(\pi)\prod_{i=1}^{2n-1}\sigma_i$. We then simply compile it into depth $\mco(\rtg\log n)$ using Theorem~\ref{thm:perm_comp}. By the discussion of Section~\ref{sec:prelim}, the fact that our 3-design elements themselves belong to the matchgate group allows us to automatically conclude that they form a \textit{strong} matchgate 3-design. Finally, the fact that logarithmic depth is optimal follows from an elementary light   cone argument discussed in Section~\ref{sec:discussion}.
\end{proof}

\noindent
Corollary~\ref{crl:threedes} is very convenient: we achieve a zero-error 3-design (which is already sufficient for many applications~\cite{huang2020predicting,west2026classical,wan2022matchgate,zhao2021fermionic,low2022classical,west2026particle}) using only Clifford gates, which are particularly amenable to be implemented within many quantum error correcting codes~\cite{gottesman1998theory,bombin2015gauge,steane1996multiple,bombin2006topological}. \\

\noindent
For  general $k>3$, however, we can no longer rely on sampling random matchgate Cliffords, and instead need to implement the three step program described at the start of this section. We find:\\

\begin{restatable}{thm}{thmmain}\label{thm:main}
There is a strong  $\varepsilon$-approximate relative   error $n$-qubit matchgate $k$-design consisting of  matchgates which can be compiled into a 2-local circuit of depth $\mco(k^2\rtg\log n\,\log (n/\varepsilon))$. For all-to-all connectivity,   the compilation is efficient.
\end{restatable}
\begin{proof}
Recall that our proof has, at a high level, three steps. We begin with Step 1, in which we introduce   the random walk which will constitute the elements of our  designs, and analyse its spectral properties. This analysis will allow us to determine a  sufficient walk length to achieve an approximate design of a given target error. \\

\noindent
\textit{Step 1: Bounding the walk length.} We will construct our approximate designs from  random walks on the matchgate group,    taking a  single step   of the random walk to have, for $a<b$, the form $R_{ab}(\vartheta)=\exp\left(-\vartheta c_ac_b/2\right)$. Here both the pair $a<b$ and the rotation angle $\vartheta$ are chosen uniformly randomly.  We recall from Section~\ref{sec:prelim} that the Hermitian generator $X_{ab}$ is defined implicitly by $\Ad_{R_{ab}(\vartheta)\tk}=e^{-i\vartheta X_{ab}}$, and that we say a  vector has charge $q$ exactly  when it is an eigenvector of $X_{ab}$ with eigenvalue $q$. Now, for a fixed $(a,b)$, the uniform average over the rotation angle preserves the
charge-zero subspace and annihilates everything else, so that $\expect_\vartheta[\Ad_{R_{ab}(\vartheta)\tk}]=\mathbf1_{\{X_{ab}=0\}}=:P_{ab}$, where $P_{ab}$ is the projector onto the kernel of $X_{ab}$. A single step of our random walk therefore has $k$\textsuperscript{th} moment operator
\begin{equation}
    S(A) = \expect_{R\sim\nu}\big[R\tk A (R\ad)\tk\big] = \frac{1}{\binom{2n}{2}} \sum_{a<b}\frac{1}{2\pi}\int {\rm d}{\vartheta}\ R_{ab}(\vartheta)\tk A (R_{ab}(\vartheta)\ad)\tk = \frac{1}{\binom{2n}{2}}\sum_{a<b}P_{ab}(  A )\,,
\end{equation}
where we denote by $\nu$ the measure on the matchgate group that captures our setup (i.e., pick a random pair, a random angle, and apply the corresponding $R_{ab}(\vartheta)$). As discussed, we are interested in analysing the spectral properties of $S$, which it will be helpful to do   on an irrep-by-irrep basis. We let $\Sigma^{(k)}$ denote the set of distinct
irrep types that appear in the $k$-fold tensor power of the adjoint representation of the matchgate group, 
and for each $\sigma\in\Sigma^{(k)}$, use
$\pi_\sigma:\mg\to\mbu(V_\sigma)$ to denote the irrep, let $d_\sigma=\dim V_\sigma$, and define $S_\sigma:=\int_{g\sim \nu}\pi_\sigma(g)$. 
That is, $S_\sigma$ is the block of $S$ acting on each copy of $V_\sigma$. We now have a useful technical lemma:


\begin{restatable}{lem}{leminfbound}\label{lem:inf_bound}
The nontrivial  $\sigma\in\Sigma^{(k)}$ can be partitioned into sets $\Sigma_r^{(k)}$, with $1\leq r\leq nk$, such that $\sum_{\sigma\in\Sigma_r}d_\sigma^2\leq(2n)^{2r}$, and, for every $\sigma\in\Sigma_r^{(k)}$,
\begin{equation*}
    \|S_\sigma\|\leq \exp \left(-\frac{r}{k^2(2n-1)}\right).
\end{equation*}
\end{restatable}

\noindent
We postpone the proof of Lemma~\ref{lem:inf_bound}, which is just an exercise in  representation theory, to Appendix~\ref{sec:fm}. The utility of the spectral bound it provides is revealed by the next lemma:

\begin{restatable}{lem}{lemfourier}\label{lem:fourier}
If 
\begin{equation}
 \sum_{\substack{\sigma\in\Sigma\\\sigma\neq\mathbf1}}d_\sigma^2\|S_\sigma\|^m\leq\delta \label{eq:fourier_bound}
\end{equation}
for some $m\geq 1$ and $0<\delta < 1$, then the $m$-step  measure $\nu^{*m}$ is a $\delta$-approximate relative error
design.
\end{restatable}
Let us provide some intuition as to what is going on in Lemma~\ref{lem:fourier} (again postponing the actual proof to the appendix). The idea is that, when  decomposed into  matchgate-irreps, the action of some $g\in\mg$  looks like
\begin{equation*}
\Ad_{g\tk}= {\rm diag}(1,1,\pi_\sigma(g),\pi_\sigma(g),\pi_\tau(g),\ldots)
\end{equation*}
where the first two\footnote{The multiplicity of the trivial irrep being two here is an arbitrary choice made for the sake of illustration}   blocks are copies of the trivial irrep, followed by the various other appearing irreps, with their multiplicities. Averaging $g$ over $\nu$ gives 
\begin{equation*}
    S= \expect_{g\sim \nu} \Ad_{g\tk} = {\rm diag}(1,1,S_\sigma,S_\sigma,S_\tau,\ldots);
\end{equation*}
for an $m$-step walk we have
\begin{equation}\label{eq:blocks}
    S^m= \expect_{g\sim \nu^{*m}} \Ad_{g\tk} = {\rm diag}(1,1,S_\sigma^m,S_\sigma^m,S_\tau^m,\ldots).
\end{equation}
On the other hand, the corresponding average  with respect to the Haar measure simply reads 
\begin{equation}\label{eq:blocks_haar}
\expect_{g\sim\mu_{\mg}}\Ad_{g\tk}={\rm diag}(1,1,0,0,0,0...),    
\end{equation}
as the non-trivial irreps average exactly to zero by the standard orthogonality theorems~\cite{fulton1991representation}. As we are ``trying to look Haar random'', we see that the quantity on the left-hand side of Eq.~\eqref{eq:fourier_bound} is clearly something we would like to be small (so that the right hand sides of Eqs.~\eqref{eq:blocks} and~\eqref{eq:blocks_haar} approximate each other); the content of Lemma~\ref{lem:fourier} is then that it is in fact   directly related to the   approximation error with respect to the relative error metric. \\

At this point, we have good control over the spectrum of a single step of the random walk of our ensemble (Lemma~\ref{lem:inf_bound}) and a spectral criterion for forming a relative error design (Lemma~\ref{lem:fourier}); combining them  swiftly yields

\begin{restatable}{lem}{lemwalkdepth}\label{lem:walk_depth}
If the number $m$ of steps in our walk satisfies
\begin{equation}
 m\geq k^2(2n-1)\log\!\left(\frac{16n^2}{\delta}\right),
\end{equation}
then $\nu^{*m}$ is a strong  $\delta$-approximate relative error matchgate
$k$-design.  
\end{restatable}

\noindent
This concludes the first step of our proof program.\\

\noindent
\textit{Step 2: Bounding the number of fermionic layers.} The second step is about controlling the length of the circuits produced in Step 1; as those circuits consist of $\mco(k^2 n\log (n/\varepsilon))$ 2-Majorana rotations,  they could in principle be far deeper than the logarithmic depth we are targeting; the question is how much parallelism is possible. So, let us analyse our ability to compress these circuits into layers of commuting rotations. Writing the   random pairs of an $m$-step walk as
$e_1,\ldots,e_m$, we assign to the 
 $i$\textsuperscript{th} rotation the level
\begin{equation}
 \ell_i=1+\max\{\ell_j:j<i,\ e_j\cap e_i\neq\{\}\},
 \label{eq:levels}
\end{equation}
use the convention $\max(\{\})=0$, and put $D=\max_i\ell_i$. That is,   a gate is placed one level after every
earlier gate that uses either of its Majorana coordinates (this is a slightly conservative scheme, as if the same pair is sampled more than once, then those gates will commute with each other; for simplicity, we ignore this, at the cost of potentially making our circuits very slightly deeper than they need to be). Analysing this, we find

\begin{restatable}{lem}{lemparadepth}\label{lem:para_depth}
 There is a constant $C$ such that, for $0<\eta<1$,   $L= C \max\left\{m/n,\, \log (n/\eta)\right\}$ implies $\Pr[D\geq L]\!<\!\eta$. 
\end{restatable}

So, Lemma~\ref{lem:para_depth} gives us hope   that, even though $m\in\mco(k^2 n\log (n/\varepsilon))$, we will be able to compress things down to logarithmic depth (although at this point there remains a bit to be seen before we can draw this conclusion completely).
Now, let us fix the $L$ given by Lemma~\ref{lem:para_depth} (the explicit constant is worked out in Appendix~\ref{sec:fm}), and write $\nu_{m,L}$ for the measure given by our walk conditioned on $D<L$. To sample from this measure, we draw a walk in the usual way, check if its depth is less than $L$, and if not, sample a new one (repeating as many times as necessary). While this conditioned ensemble may not itself lead to a relative error design, we are able to show that its convolution with itself does. Operationally, this convolution is just the concatenation of two such walks, so that the number of layers of the resulting circuit is bounded by $2(L-1)$. Formally, we find
\begin{restatable}{lem}{lemmajodoubled}\label{lem:doubled}
Let $0<\varepsilon\leq1$, and set $\delta=\eta=\varepsilon/3$. Taking $m$ as in Lemma~\ref{lem:walk_depth} and $L$ as in Lemma~\ref{lem:para_depth}, the ensemble  
$\mu_{m,L}:=\nu_{m,L}*\nu_{m,L}$ has strong relative error at most $\varepsilon$. 
\end{restatable}

The delivery unto us by Lemma~\ref{lem:doubled} of an ensemble consisting of $2L\in\mco\left(k^2\log (n / \varepsilon)\right)$ layers of commuting rotations concludes Step 2.\\

\textit{Step 3: Bounding the physical depth of each layer.} At this point, we have an ensemble  of circuits constituting an $\varepsilon$-approximate matchgate design comprised of $\mco(k^2\log(n/\varepsilon))$ layers, with each layer formed from  elementary gates of the form $R_{ab}(\vartheta)=\exp\left(-\vartheta c_ac_b/2\right)$, and with every Majorana $c_j$ appearing at most once in each layer (this follows from Eq.~\eqref{eq:levels}).  
Such ``elementary gates'', however, may be of extensive physical support, and it remains to be seen that each layer can be compiled into a logarithmic depth circuit of physical 2-qubit gates; this is the task of Step 3. 
Happily, this turns out to be a fairly immediate corollary of Theorem~\ref{thm:perm_comp}. First, note that  if it happens that $a=2j-1,\ b=2j$ for some $j$, then $\exp(-\vartheta c_ac_b/2)=\exp(-i\vartheta Z_j/2)$ is of course easy to implement; more generally we can apply a Majorana   permutation $Q$ to move the pairs of Majoranas ``next to each other'' in this fermionic sense, at which point their joint rotation simply becomes an easy single-qubit rotation. Concretely, for  a $V$ layer involving pairs $\{(a_j,b_j)\}_{j=1}^s$, let  $Q$  send each $(a_j,b_j)$ to $(2j-1,2j)$, so that
\begin{equation*}
    V=\prod_{j=1}^{s}R_{a_jb_j}(\vartheta_j) 
= Q\ad\Big[\prod_{j=1}^{s} e^{-i\vartheta_j Z_j/2}\Big]Q.
\end{equation*}
By Theorem~\ref{thm:perm_comp} we can compile $Q$, and hence each layer, into a circuit of depth $\mco(\rtg\log n)$; as there are $\mco(k^2\log(n/\varepsilon))$ total layers, this   yields the statement of the theorem. 

\end{proof}

\section{Discussion}\label{sec:discussion}
Compared to known constructions~\cite{west2025no,grevink2025will,west2026ambient,braccia2025optimal}, our setup exponentially reduces the depths required to form matchgate $k$-designs, at the cost of requiring all-to-all connectivity; it is natural to wonder, however, whether our scaling is in fact optimal. To that end, one can obtain a simple lower bound on the required depths using a lightcone argument. To see this, let us use two (strong) queries to prepare $(U\otimes\overline U)\sdket{X_1}=\sdket{UX_1U\ad}$, where $\sdket{X_1}=(X_1\ot\id_\mch)\ket\Phi$ is the Choi state of the Pauli $X$ operator on the first qubit. Now, for any depth-$D$ 2-local circuit, $UX_1U^\dagger$ is supported on at most $2^D$ qubits, with its Pauli expansion therefore containing only strings of weight at most $2^D$. Consider, then, measuring the projector onto all vectorised Majoranas strings  supported on more   than $2^D$ qubits: Every   circuit of depth less than $D$ will accept with probability zero, but  under the action of a matchgate 2-design, we will accept with probability $\max\left\{0,\,1-{2^D}/{n}\right\} $. It follows that  we require depth at least $\Om(\log n)$ to form even a matchgate 2-design, so that for $k=2,3$ our construction is optimal. For generic $k>3$, though, there remains a factor of $\log n$ between this  lower bound and our  construction.\\

It is also interesting to wonder whether our $k$-dependence can be improved to polylogarithmic, while retaining the polylogarithmic dependence on $n$. We show in Appendix~\ref{sec:kdep} this is impossible, by deriving, for $1\leq k\leq n$, a general lower bound on the depth of ancillae-free relative error matchgate designs given by $\Om(k/\log(nk))$. This does in principle leave a little room for improvement in our $k$-dependence: For  $k\sim n^{1/a}$ with $a>1$, this gives a stronger lower bound than the naive light cone argument of the previous paragraph. 
Interestingly, as it is known that the Haar measure on the matchgate group can be sampled from exactly in depth $\mco(n)$ (even under the assumption of one dimensional connectivity~\cite{braccia2025optimal}), as $k$ becomes comparable to $n$, our lower bound approaches the depth required for exact Haar sampling, up to logarithmic factors. Thus, the potential depth advantage of all-to-all connectivity is reduced to at most a logarithmic factor.  \\

Finally, we remark that our results have implications for  various popular algorithms that involve sampling from approximate matchgate designs. When only a 3-design is needed, we can by Corollary~\ref{crl:threedes} do this using only Clifford gates, and therefore transversally in many fault tolerant quantum computation schemes~\cite{gottesman1998theory,bombin2015gauge,steane1996multiple,bombin2006topological}. A prototypical such example is that of performing fermionic classical shadows~\cite{wan2022matchgate,zhao2021fermionic,low2022classical,west2026particle}, the previously required linear circuit depth  of which  may then be exponentially improved on quantum computers with all-to-all connectivity compared to that strictly one dimensional counterparts. 
As such connectivity is natural in trapped ion~\cite{moses2023race,ransford2025helios} and neutral atom~\cite{bluvstein2022quantum,bluvstein2024logical} processors, such systems may become the natural architecture for these fermionic algorithms.

\section*{Acknowledgements}
We acknowledge useful discussions with Lukasz Cincio.  MW, MC and ML acknowledge support by the Laboratory Directed Research and Development (LDRD) program of Los Alamos National Laboratory (LANL) under project number 20260043DR, and by LANL’s ASC Beyond Moore’s
Law project. This work was also supported by the Quantum Science Center, a National Quantum Information Science Research Center of the U.S. Department of
Energy. The technical details of the proofs of our results were largely proposed by a combination of GPT-5.6 and GPT-6, though with non-zero human input. All AI-generated results were verified by the authors, and variously rewritten in order to meet our standards of clarity.

\bibliography{quantum}

\appendix

\clearpage
\newpage

\section{Optimal fermionic routing}\label{sec:perm}
\noindent
In this appendix we prove Theorem~\ref{thm:perm_comp}, which we recall for convenience to state
\thmpermcomp*
\begin{proof}
For now, let us assume all-to-all connectivity; we will deal with the general case at the end. 
We find it useful to begin by efficiently   compiling permutations whose action is only specified   on the Majoranas of odd index; to simplify the notation it is then convenient to introduce $a_j=c_{2j-1},\ b_j=c_{2j}$, for $1\leq j \leq n$. 
Our proof that we can efficiently compile permutations amongst these $a_i$ (without worrying for now where the even Majoranas are sent) will be graph theoretical. Indeed, let us introduce a binary tree which in the case of $n=3$ qubits takes the following form (with the obvious generalisation to more or fewer qubits):
\begin{center}
\refstepcounter{graph}\label{gr:graph}
\begin{tikzpicture}[x=1cm,y=1cm]
\node[qubit] (q1) at (0,0) {$1$};
\node[qubit] (q2) at (0,-1.45) {$2$};
\node[qubit] (q3) at (0,-2.90) {$3$};
\node[leaf] (a1) at (-2.1,0) {$a_1=X_1$};
\node[leaf] (b1) at (2.1,0) {$b_1=Y_1$};
\node[leaf] (a2) at (-2.1,-1.45) {$a_2=Z_1X_2$};
\node[leaf] (b2) at (2.1,-1.45) {$b_2=Z_1Y_2$};
\node[leaf] (a3) at (-2.1,-2.90) {$a_3=Z_1Z_2X_3$};
\node[leaf] (b3) at (2.1,-2.90) {$b_3=Z_1Z_2Y_3$};
\node[leaf] (p) at (0,-4.0) {$P=Z_1Z_2Z_3$};
\foreach \j in {1,2,3}{
 \draw[edge] (q\j)--node[port,above] {$X$} (a\j);
 \draw[edge,spectator] (q\j)--node[port,above] {$Y$} (b\j);
}
\draw[edge] (q1)--node[port,right] {$Z$} (q2);
\draw[edge] (q2)--node[port,right] {$Z$} (q3);
\draw[edge] (q3)--node[port,right] {$Z$} (p);

\node (p) at (8.,.-1.45)  {\textcolor{C4}{(\thegraph)}};
\node (p) at (-8.,.0) { };

\end{tikzpicture}
\end{center}
We emphasise that the orange edges are for our near-term purposes essentially fictitious, being merely along for the ride; the ``true'' leaves of the tree consist exactly of the $a_i$ and the parity operator $P=\prod_j Z_j$, and its edges are given exactly by those in black in the figure. Notice that to each leaf we associate an operator which is given by tracing a path from   the root to the leaf, applying along the way the Pauli recorded on an edge to the qubit which was last seen along the path. We will be interested in operating on these graphs via \textit{Whitehead moves}~\cite{rafi2011diameter}. A Whitehead move can be thought of as contracting an edge $e$ and then expanding it back out again in a different way; graphically, the possible moves look like (c.f. Figure 2 of Ref.~\cite{rafi2011diameter})

\begin{center}
   
\begin{tikzpicture}[
  scale=0.75,
  transform shape,
  x=1.65cm,y=1.2cm,
  wh-edge/.style={draw=black,line width=1.05pt},
  wh-junction/.style={circle,draw=black,fill=black!50!white,line width=0.9pt,inner sep=0pt,
    minimum size=4.4mm},
  wh-end/.style={circle,draw=black,line width=0.9pt,fill=blue!60!green!60!white,
    inner sep=0pt,minimum size=4.4mm},
  wh-label/.style={font=\Large,inner sep=3pt},
  wh-pairing/.style={font=\normalsize,inner sep=2pt},
  wh-arrow/.style={-{Stealth[length=3mm,width=2.2mm]},line width=1pt}
]
\begin{scope}
  \coordinate (u) at (-0.5,0);
  \coordinate (v) at (0.5,0);
  \coordinate (a) at (-1.35,0.85);
  \coordinate (b) at (-1.35,-0.85);
  \coordinate (c) at (1.35,0.85);
  \coordinate (d) at (1.35,-0.85);
  \draw[wh-edge] (a)--(u)--(b);
  \draw[wh-edge] (c)--(v)--(d);
  \draw[wh-edge] (u)--node[wh-label,below] {$e$}(v);
  \foreach \p in {u,v} \node[wh-junction] at (\p) {};
  \foreach \p in {a,b,c,d} \node[wh-end] at (\p) {};
  \node[wh-label,above=3mm] at (a) {$a$};
  \node[wh-label,below=3mm] at (b) {$b$};
  \node[wh-label,above=3mm] at (c) {$c$};
  \node[wh-label,below=3mm] at (d) {$d$};
  \node[wh-pairing] at (0,1.) {$\{a,b\}\mid\{c,d\}$};
\end{scope}

\draw[wh-arrow] (2.1,1.5)--(2.95,1.8);
\draw[wh-arrow] (2.1,-1.5)--(2.95,-1.8);

\begin{scope}[shift={(5,2)}]
  \coordinate (u1) at (-0.5,0);
  \coordinate (v1) at (0.5,0);
  \coordinate (a1) at (-1.35,0.85);
  \coordinate (c1) at (-1.35,-0.85);
  \coordinate (b1) at (1.35,0.85);
  \coordinate (d1) at (1.35,-0.85);
  \draw[wh-edge] (a1)--(u1)--(c1);
  \draw[wh-edge] (b1)--(v1)--(d1);
  \draw[wh-edge] (u1)--node[wh-label,below] {$e'$}(v1);
  \foreach \p in {u1,v1} \node[wh-junction] at (\p) {};
  \foreach \p in {a1,b1,c1,d1} \node[wh-end] at (\p) {};
  \node[wh-label,above=3mm] at (a1) {$a$};
  \node[wh-label,below=3mm] at (c1) {$c$};
  \node[wh-label,above=3mm] at (b1) {$b$};
  \node[wh-label,below=3mm] at (d1) {$d$};
  \node[wh-pairing] at (0,1.) {$\{a,c\}\mid\{b,d\}$};
\end{scope}

\begin{scope}[shift={(5,-2)}]
  \coordinate (u2) at (-0.5,0);
  \coordinate (v2) at (0.5,0);
  \coordinate (a2) at (-1.35,0.85);
  \coordinate (d2) at (-1.35,-0.85);
  \coordinate (b2) at (1.35,0.85);
  \coordinate (c2) at (1.35,-0.85);
  \draw[wh-edge] (a2)--(u2)--(d2);
  \draw[wh-edge] (b2)--(v2)--(c2);
  \draw[wh-edge] (u2)--node[wh-label,below] {$e''$}(v2);
  \foreach \p in {u2,v2} \node[wh-junction] at (\p) {};
  \foreach \p in {a2,b2,c2,d2} \node[wh-end] at (\p) {};
  \node[wh-label,above=3mm] at (a2) {$a$};
  \node[wh-label,below=3mm] at (d2) {$d$};
  \node[wh-label,above=3mm] at (b2) {$b$};
  \node[wh-label,below=3mm] at (c2) {$c$};
  \node[wh-pairing] at (0,-1.) {$\{a,d\}\mid\{b,c\}$};
\end{scope}
\end{tikzpicture}
\end{center}

Fixing the edge towards one of the vertices (say $a$) as pointing ``towards the root'' of the tree,  one of the subtrees  shown here is then determined   by which of the other vertices is ``partnered with'' $a$, with the Whitehead moves   changing who that partner is.
Now, we would like to define an action of the Clifford group on our tree. The rule is simple: to each leaf of the tree is associated a Pauli as described above; we let a Clifford act on the tree by acting (in the adjoint representation) on each leaf of the tree simultaneously. Let us see that   Whitehead moves correspond  exactly to the action of CNOT and Hadamard gates (acting on the qubits whose vertices are at the ends of the edge which is contracted and re-expanded)!   So, fix  two qubits, $u$ and $v$, who are currently neighbours in the tree. One of them (say $u$) will be higher up in the tree, with one of its two children being $v$, and the other some subtree (call it $D$). $v$ itself has two children, $E$ and $F$. In general, these  children descend from their parent qubit via either an ``$X$'' edge or a ``$Z$'' edge; let us assume that the descent from $u$ to $v$ is along the $Z$ edge (if it is not, we can simply act on the tree with $H_u$, a Hadamard on qubit $u$; this Hadamard will also introduce a phase on the corresponding spectator Majorana, which we ignore for now). Now, we can see directly that applying ${\rm CNOT}_{u\to v}$ acts on the relevant subtree as:
\begin{center}
    \begin{tikzpicture}[x=0.9cm,y=0.9cm]
\begin{scope}
\node[qubit] (u) at (0,0) {$u$};
\node[qubit] (v) at (1,-1.4) {$v$};
\node[leaf] (D) at (-1,-1.4) {$D$};
\node[leaf] (E) at (0.1,-2.8) {$E$};
\node[leaf] (F) at (1.9,-2.8) {$F$};
\node[leaf,text=orange!65!black] (yu) at (1.2,0.15) {$y_u$};
\node[leaf,text=orange!65!black] (yv) at (2.4,-1.25) {$y_v$};
\draw[edge,-{Stealth}] (u.north)--++(0,0.65)
  node[above,font=\scriptsize] {toward root};
\draw[edge] (u)--node[port,left] {$X$}(D);
\draw[edge] (u)--node[port,right] {$Z$}(v);
\draw[edge] (v)--node[port,left] {$X$}(E);
\draw[edge] (v)--node[port,right] {$Z$}(F);
\draw[edge,spectator] (u)--node[port,above] {$Y$}(yu);
\draw[edge,spectator] (v)--node[port,above] {$Y$}(yv);
\end{scope}
\draw[-{Stealth},thick] (3.0,-1.25)--node[above,font=\small] {${\rm CNOT}_{u\to v}$}(6.3,-1.25);
\begin{scope}[xshift=8.4cm]
\node[qubit] (v2) at (0,0) {$v$};
\node[qubit] (u2) at (-1,-1.4) {$u$};
\node[leaf] (D2) at (-1.9,-2.8) {$D$};
\node[leaf] (E2) at (-0.1,-2.8) {$E$};
\node[leaf] (F2) at (1,-1.4) {$F$};
\node[leaf,text=orange!65!black] (yv2) at (1.2,0.15) {$y_v$};
\node[leaf,text=orange!65!black] (yu2) at (-2.4,-1.25) {$y_u$};
\draw[edge,-{Stealth}] (v2.north)--++(0,0.65)
  node[above,font=\scriptsize] {toward root};
\draw[edge] (v2)--node[port,left] {$X$}(u2);
\draw[edge] (v2)--node[port,right] {$Z$}(F2);
\draw[edge] (u2)--node[port,left] {$X$}(D2);
\draw[edge] (u2)--node[port,right] {$Z$}(E2);
\draw[edge,spectator] (v2)--node[port,above] {$Y$}(yv2);
\draw[edge,spectator] (u2)--node[port,above] {$Y$}(yu2);
\end{scope}
\end{tikzpicture}
\end{center}
Indeed, note that
\begin{align*}
\Ad_{{\rm CNOT}_{u\to v}}(X_u)  &=X_uX_v\\
\Ad_{{\rm CNOT}_{u\to v}}(Y_u)   &=Y_uX_v\\
\Ad_{{\rm CNOT}_{u\to v}}(Z_uX_v)  &=Z_uX_v\\
\Ad_{{\rm CNOT}_{u\to v}}(Z_uY_v)  &=Y_v\\
\Ad_{{\rm CNOT}_{u\to v}}(Z_uZ_v)  &=Z_v
\end{align*}
To be clear about what   just happened, let us assume that we have only two qubits, and that $u,v,D,E$ and $F$ are respectively  $1,2,a_1,a_2$, and $P$ (seeing Graph~\ref{gr:graph}). Acting with ${\rm CNOT}_{1\to 2}$   then (for example) maps $D=a_1$ to a vertex which is linked to the root by $X_2$ and $X_1$, which is exactly how ${\rm CNOT}_{1\to 2}$ acts on $a_1=X_1$; one can check that all the other vertices are indeed mapped where they should be. So, we can implement either   Whitehead moves on a given pair of vertices by optionally applying $H_v$, and then applying a CNOT. \\

Next, by Theorem   3.1 of Ref.~\cite{rafi2011diameter}, we can (efficiently) find a sequence of length $\mco(\log n)$ of \textit{simultaneous Whitehead moves} that map the leaves of our tree to a tree whose leaves are ordered in any other sequence. Here a simultaneous Whitehead move is a set of Whitehead moves, no two of which involve the same vertex. In our setup, this translates to a set of CNOTs and Hadamards which can be performed in parallel (i.e., in constant depth). This means that, for a target permutation $\pi$ of the odd Majoranas, we can after a logarithmic number of steps modify Graph~\ref{gr:graph} so that the vertices on the left hand side are ordered according to $\pi$ (we in general need a final layer of Hadamards acting on some of the qubits to get all of the $X$ and $Z$ edges pointing in the manner of Graph~\ref{gr:graph}). In the course of our operations, we will have as a byproduct moved around in some fashion the ``qubit vertices'', who will now be ordered according to some (generically different) permutation $\sg$;  the fictitious $b$ vertices will have come along for the ride, being now also ordered according to $\sg$, potentially with some minus signs introduced by the various Hadamards.\\

Now, the permutation $\sg^{-1}$ can be applied (acting in the usual qubit-permuting  representation) in constant depth; indeed it is an elementary property of permutations that they can be written as the composition of two permutations, each of which consists only of transpositions (and fixed points). Under our assumption of all-to-all connectivity, then, this is just a depth-2 layer of two-qubit SWAPs. What is the point of doing that? Well, consider for example the $j$\textsuperscript{th} vertex on the left hand side of Graph~\ref{gr:graph} (after having done the Whitehead moves but before these final swaps). It now has a label which will read    $a_{\pi^{-1}(j)}$,   recording the fact that if one were to apply  to  $a_{\pi^{-1}(j)}$ the various sequence of CNOTs and Hadamards that constitute the applied Whitehead moves, one would obtain the Pauli $Z_{\sg(1)}Z_{\sg(2)}\cdots Z_{\sg(j-1)}X_{\sg(j)}$. But this is not what we want (indeed, in the Jordan-Wigner sense, it is probably not even a Majorana)! What we want is rather to map the $(\pi^{-1}(j))$\textsuperscript{th} odd Majorana to the $j$\textsuperscript{th} odd Majorana. But to do this, we simply need to  apply $\sg^{-1}$ in the qubit-permuting  representation. \\

At this point, we can apply in log-depth a permutation that maps the odd numbered Majoranas amongst themselves in a matter of our choosing, at the cost of permuting the even numbered Majoranas amongst themselves in an unknown way, and additionally multiplying some of them by minus one. A natural next step then, is to achieve a permutation amongst the odd Majoranas while leaving the even Majoranas untouched. To that end, we recall the elementary fact that any permutation may be written as the composition of two involutions, $\pi=\tau_1\tau_2$, with $\tau_1^2=\tau_2^2=e$, the identity permutation. For a target permutation $\pi\in \mathfrak{S}_n$, write such an involution as $\tau=\prod_r (u_r\,v_r)$, and choose a further permutation $\rho\in \mathfrak{S}_n$ such that $\rho(2r-1)=u_r$ and $\rho(2r)=v_r$ for all the appearing pairs $(u_r\,v_r)$; never mind where $\rho$ sends the fixed points of $\tau$. Next, introduce 
\begin{equation*}
K=\prod_r\exp\left(-\frac{\pi}{4}a_{2r-1}a_{2r}\right)=\prod_r\exp\left(\frac{i\pi}{4}Y_{2r-1}X_{2r}\right) ,
\end{equation*}
which  evidently may be implemented in constant depth. Finally, with $F_\rho$ the previously constructed log-depth circuit for implementing the permutation $\rho$ on the $a_i$ and acting in some unknown way  on the $b_i$, notice that the map
\begin{equation}
F_\rho K F_\rho^\dagger=\prod_r\exp\left(-\frac{\pi}{4}a_{u_r}a_{v_r}\right).\label{eq:inv}
\end{equation}
sends each $a_{u_r}\mapsto a_{v_r}$,  and 
$a_{v_r}\mapsto-a_{u_r}$, while fixing every other Majorana (in particular, every $b_j$ is fixed). So, we have (up to some known minus signs) implemented the involution $\tau$ in log depth; repeating for the second involution then yields $\pi$, while leaving every $b_j$  fixed. Similarly, we can get odd-family-fixing  even-family permutations. Indeed, as $S^{\otimes n}$ satisfies $S\tn a_j (S\tn)\ad = b_j$, and $S\tn b_j (S\tn)\ad = -a_j$, we can simply conjugate by $S$-gates and then run the odd permutation construction.  At last, consider an arbitrary permutation of the Majoranas. Some number, say $t$, of the odd Majoranas want to become even, and the same number of even Majoranas want to become odd. Begin by   mapping those parity-changing Majoranas to the first $t$ slots of their respective families, and then apply $S^{\ot t}$. At this point, every Majorana is in its correct family, and we can use the within-family construction to finish the compilation (up to minus signs). \\

Penultimately, let us correct the minus signs introduced by the action of the $S$-gates, and that of Eq.~\eqref{eq:inv} (notice that the minus signs introduced earlier by the Hadamard gates   all cancelled out). The final trick is to notice that conjugation by the operator $Pc_j$ (where recall $P=Z\tn$ is the parity operator) flips exactly the sign   of $c_j$, leaving all other Majoranas fixed. So, we can just consider the product of $Pc_j$ over all of the Majoranas whose signs we wish to flip; this is a single Pauli string, which can then be applied in constant depth, concluding the construction. \\

Finally, let us consider the effect of working on an arbitrary qubit connectivity graph. What we have shown so far is that, under
all-to-all connectivity, we have a router of depth $\log n$ (which is optimal by the usual lightcone arguments; indeed, consider a permutation which swaps $X_1$ with $Z_1Z_2\ldots Z_{n-1}X_n$). 
Now, from the work of Ref.~\cite{yuan2025full}, it follows that the compilation onto an arbitrary graph $G$ incurs only a factor of $\rtg$; moreover, the compilation process simply involves inserting 2-qubit SWAP gates in strategic positions. As SWAP gates are also Clifford, the claim follows.\\
\end{proof}

\section{The impossibility of polylog$(n,k)$-dependence}~\label{sec:kdep}

\noindent
In this appendix we establish the lower bound $D\geq \Om(k/\log(nk))$  on the depth of ancilla-free relative error matchgate $k$-designs claimed in Section~\ref{sec:discussion}. \\

\noindent
Let $v=\ket{0}\tn\in \mch^+$ be the highest-weight vector (of weight $\om=\tfrac 12 (1,1,\ldots ,1)$) in the even-parity irrep which occurs in the standard representation $\mch\cong \mch^+\oplus\mch^-\cong(\mbc^2)\tn$ of the matchgate group, and define, for some ensemble $\mu$ consisting of circuits of depth at most $D$,
\begin{align*}
\rho_\mu&=\expect_{U\sim\mu}\bigl(\ketbra{Uv}\bigr)^{\otimes k}\\
\rho_G&=\expect_{U\sim\mu_G}\bigl(\ketbra{Uv})^{\otimes k}.
\end{align*}
Recalling that $\mu$ being an $\varepsilon$-approximate relative error   $k$-design over a group $G$ implies $(1-\varepsilon)\Phi_{\mu_G}^{(k)}
 \preceq\Phi_\mu^{(k)}$ (and also an additional upper bound that we will not need here), we see that we must have (for $\varepsilon<1$) that $\rank(\rho_\mu)\geq \rank(\rho_G)$. Indeed, suppose that $u$ is in the kernel of $\rho_\mu$. Then we must have 
 \begin{equation*}
   0\leq  u\ad\big[   \rho_\mu - (1-\varepsilon)\rho_G \big]u = -(1-\varepsilon)u\ad \rho_G u,
 \end{equation*}
 so that by the positive semidefiniteness of  $\rho_G$ we must have $u\in\ker \rho_G$. So, $\ker \rho_\mu \subseteq \ker \rho_G$, whence $\rank(\rho_\mu)\geq \rank(\rho_G)$. 
 The idea of the proof will just be to compute these ranks, and see that the   inequality between them will  turn out to force a lower bound on $D$. Let us begin with $\rank(\rho_G)$. \\

 Now,   $v^{\otimes k}$ is a highest-weight vector of weight $k\omega$, whose  orbit spans the irrep $V_{k\omega}$, so that  Schur's lemma therefore gives
$\rho_G=\id_{V_{k\omega}}/\dim V_{k\omega}$. We then have~\cite{fulton1991representation}
\begin{equation}
\rank\rho_G =\dim V_{k\omega}=\prod_{1\leq i<j\leq n} \left(1+\frac{k}{2n-i-j}\right) \geq\left(1+\frac{k}{2n}\right)^{n(n-1)/2}. \label{eq:haar_rank}
\end{equation}

On the other hand, fix some element $U$ of the ensemble $\mu$, which consists of (say)  $s$ one- or two-qubit gates. Let us regard each gate's matrix entries as independent complex variables (of course unitarity imposes some constraints, we are being conservative for simplicity); there are then at most $16$ per gate. Every coordinate of $(Uv)^{\otimes k}$ is then a  homogeneous polynomial of degree $k$ in the entries of
each gate separately. Indeed, to be explicit, start with a single two-qubit gate $T=\sum_{a=1}^{16}x_a E_a$,  where the $E_a$ are the matrix units, and let $w_a=E_av$. Then
\begin{align*}
(Tv)^{\otimes k}&=\left(\sum_{a=1}^{16}x_aw_a\right)^{\otimes k}= \sum_{a_1,\ldots,a_k}x_{a_1}\cdots x_{a_k}\,w_{a_1}\otimes\cdots\otimes w_{a_k}= \sum_{\substack{\alpha_1,\ldots,\alpha_{16}\geq0\\\alpha_1+\cdots+\alpha_{16}=k}}x_1^{\alpha_1}\cdots x_{16}^{\alpha_{16}}\,W_\alpha,
\end{align*}
where $W_\alpha$ is the (independent of $T$) sum of the corresponding tensor products of the $w_a$. There are $B_k=\binom{k+15}{15}$ possible ``multiplicity vectors'' $\alpha$, so that for   a different two-qubit gate $T'$ (but still with the same support as $T$),  $(T'v)^{\otimes k}$  lies in the span of those same $B_k$  vectors. In the general case of $s$ gates, we similarly conclude that (having fixed $s$ locations at which the gates are to be applied) $(Uv)\tk=((T_s\cdots T_1)v)\tk$ always lies in the span of some   $B_k^s$  vectors, no matter the chosen $T_i$ in total. Now, there are $\mco(n^2)$ choices of support per gate, leading in total to a span of dimension at most $(n^2B_k)^{s}$; finally, a depth-$D$ circuit has at most $s=nD$ gates, so that 
\begin{equation}
\rank\rho_\mu\leq(n^2B_k)^{nD}.
\end{equation}
Comparing this to Eq.~\eqref{eq:haar_rank}, we find
\begin{equation*}
D\geq \frac{(n-1)\log(1+k/(2n))} {2\bigl(2\log n+\log\binom{k+15}{15}\bigr)}\geq \frac{(n-1)k} {6n\bigl(2\log n+\log\binom{k+15}{15}\bigr)}\geq \frac{k} {12\bigl(2\log n+15\log (k+1) \bigr)}=\Om(k/\log(nk)),
\end{equation*}
where we have used that $n-1\geq n/2$, $\binom{k+15}{15}\leq(k+1)^{15}$, and that for $k\leq n$, we have
$\log(1+k/(2n))\geq k/(3n)$.

\section{Further minutiae}\label{sec:fm}

\noindent
In this appendix we supply proofs of the lemmas yet unproved.

\leminfbound*
\begin{proof}
Recall from the main text that, with   $V=\mathbb C^{2n}$, we have that $\mcl\tk:=({\rm End}\,\mch)\tk\cong (\bigwedge V)\tk$ has weights which have integer coordinates between $-k$ and $k$.
It then follows from the general    representation theory of $\mbso(2n)$ that an irrep type $\sigma$ which occurs in $\mcl\tk$ has a highest weight $\lm$ with $k\geq\lambda_1\geq\cdots\geq\lambda_{n-1}\geq|\lambda_n|$~\cite{fulton1991representation}. Set $a=(\lambda_1,\ldots,\lambda_{n-1},|\lambda_n|)$, and define its \textit{degree} by $r=\sum_i a_i$ (so that $1\leq r\leq nk$ for every nontrivial irrep). Let $\Sigma_r$ consist of the distinct irrep types of degree $r$.
Now, every irrep type in $\Sigma_r$ certainly occurs in $V^{\otimes r}$; indeed, it corresponds to an $r$-box Young tableau, and thus (via the standard construction of orthogonal group irreps through \textit{Weyl functors}~\cite{fulton1991representation}) is a subspace of $V^{\otimes r}$. Since all the distinct types in $\Sigma_r$ occur in this same tensor power, we have
\begin{equation*}
\sum_{\sigma\in\Sigma_r}d_\sigma^2 \leq \left(\sum_{\sigma\in\Sigma_r}d_\sigma\right)^2 \leq \bigl(\dim V^{\otimes r}\bigr)^2 =({2n})^{2r},
\end{equation*}
which establishes  the first claim of the lemma.
For the second claim, we begin by recalling the  quadratic Casimir, which associates to an irrep   $\lm$ the value $\langle\lm,\lm+2\rho\rangle$, where $\rho $ is (half the) sum of the positive roots, and the inner product takes place in root space~\cite{fulton1991representation}. Equipping $V$ with the standard weights $\{\pm L_i\}_{i=1}^n$, the positive roots are $\{L_i\pm L_j\}_{i<j}$~\cite{fulton1991representation}, so that $\rho = \sum_{i=1}^n(n-i)L_i$, and (with $\sg\in\Sigma_r$)
\begin{equation*}
 C_2(\sigma)   =\sum_{i=1}^n a_i(a_i+2n-2i)=nr+\sum_{i=1}^n a_i(a_i-1)+\sum_{i<j}(a_i-a_j)\geq nr.
\end{equation*}
Finally, as in the main text, let $X_{ab}$ be the Hermitian generator of rotation in the $(a,b)$-plane. Its charges in $\mcl\tk$ are integers of absolute value at most $k$. With $P_{ab}$ reprising its role from the main text as the projector onto the kernel of $X_{ab}$, we therefore have $\id-P_{ab}\geq {X_{ab}^2}/{k^2}$ (which is  immediate from considering the action of these operators on the common eigenbasis of $X_{ab}$ and $P_{ab}$).
Now, as $C_2=\sum_{a<b}X_{ab}^2$, we can  average this  inequality over $a<b$ to  obtain
\begin{equation*}
    \id-S_\sigma \succeq \frac{C_2(\sigma)}{k^2\binom {2n}2}\id \succeq \frac{r}{k^2({2n}-1)}\id.
\end{equation*}
Finally, as $S_\sigma$ is an average of orthogonal projectors, it is positive semidefinite, whence
\begin{equation*}
\|S_\sigma\|\leq 1-\frac{r}{k^2({2n}-1)}\leq\exp\left(-\frac{r}{k^2({2n}-1)}\right),    
\end{equation*}
as required.
\end{proof}

\lemfourier*
\begin{proof}
This lemma is not actually particularly specific to the matchgate group, so we keep things general.  Indeed, let $G$ be a compact group with Haar measure $\mu$, and let
\begin{equation}
 \Ad_V \cong\bigoplus_{\sigma\in\Sigma}\pi_\sigma\otimes \id_{m_\sg}.
 \label{eq:ad_decomp}
\end{equation}
be the adjoint representation corresponding to operators acting on some vector space $V$. 
Here $\Sigma$ is the set of distinct irrep types, and we have   multiplicity spaces of dimension $m_\sigma$. Now, consider any (Hermiticity-preserving)
superoperator $A:{\rm End}\, V\to {\rm End}\, V$, that respects the  block diagonalisation of Eq.~\eqref{eq:ad_decomp}, i.e.,  $ A\cong\bigoplus_{\sigma\in\Sigma}
 A_\sigma\otimes \id_{m_\sg}$. We begin by showing that, with $\beta(A)=\sum_{\sigma\in\Sigma}d_\sigma\|A_\sigma\|_1$, we have
 \begin{equation}
 -\beta(A)\Phi_{\mu_G}\preceq A\preceq\beta(A)\Phi_{\mu_G}, \label{eq:beta_cp}
\end{equation}
where $\Phi_{\mu_G}=\int_{g\sim\mu}\Ad_V(g)$. Now, let us consider the Fourier coefficients
\begin{equation}
 f_A(g)=\sum_{\sigma\in\Sigma}
 d_\sigma\operatorname{tr}\left(A_\sigma\pi_\sigma(g)^\dagger\right),
\end{equation}
which as usual satisfy
\begin{equation}\label{eq:ft}
 A=\int_{g\sim \mu} f_A(g) \Ad_V(g).
\end{equation}
Now, for every Hermitian operator $X$, both $A(X)$ and $ \Ad_V(g)(X)$
are Hermitian.  Taking the Hermitian part of
Eq.~\eqref{eq:ft} and evaluating at $X$ therefore effectively replaces
$f_A(g)$ by its real part; as  every operator can be written as a complex linear
combination of Hermitian operators, we conclude that
\begin{equation}
 A=\int_{g\sim \mu}\operatorname{Re}[f_A(g)]\Ad_V(g).
\end{equation}
Next, we note   that by  
$\lvert\operatorname{tr}(BC)\rvert\leq\|B\|_1\|C\|$, we have $\forall g$ that $|f_A(g)|\leq\sum_{\sigma\in\Sigma}
 d_\sigma\|A_\sigma\|_1=\beta(A)$; 
consequently,
\begin{equation}\label{eq:cpdecomp}
 \beta(A)\Phi_{\mu_G}\pm A=\int_{g\sim \mu}\bigl(\beta(A)\pm\operatorname{Re}f_A(g)\bigr)\operatorname{Ad}_{V(g)},
\end{equation}
where the coefficients  are nonnegative (because
$\lvert\operatorname{Re}f_A(g)\rvert\leq\beta(A)$).  But then, for any ancilla space
$ K$ and any p.s.d. operator $Y\in {\rm End}(V\ot K)$ on the system
and $ K$, tensoring   Eq.~\eqref{eq:cpdecomp} with
$\id_{ K}$ expresses the image of $Y$ as an
integral of the p.s.d operators
$(V(g)\otimes \id_{ K})Y(V(g)^\dagger\otimes \id_{ K})$ with these nonnegative weights.  We conclude that both maps $\beta(A)\Phi_{\mu_G}\pm A$ are completely positive, establishing Eq.~\eqref{eq:beta_cp}. \\

Now, take $A=\Phi_{\nu^{*m}}-\Phi_{\mu_G}$. We have (recalling Eq.~\eqref{eq:blocks}) $A_{\id}=0$, and $A_\sigma=S_\sigma^m$ for $\sg\neq\id$. From the estimate $\|S_\sigma^m\|_1\leq d_\sigma\|S_\sigma^m\|\leq d_\sigma\|S_\sigma\|^m$ we see that the hypothesis of the lemma implies that $\b(A)\leq \delta$; adding $\Phi_{\mu_G}$ to all sides of the inequality chain in Eq.~\eqref{eq:beta_cp} then establishes the lemma.
\end{proof}

\lemwalkdepth*
\begin{proof}
This is a fairly simple matter of combining Lemmas~\ref{lem:inf_bound} and~\ref{lem:fourier}. Indeed, let
\begin{equation*}
z=4n^2\exp\!\left(-\frac{m}{k^2(2n-1)}\right);
\end{equation*}
then by Lemma~\ref{lem:inf_bound} we have
\begin{equation*}
\sum_{\substack{\sigma\in\Sigma\\\sigma\neq\mathbf1}}
 d_\sigma^2\|S_\sigma\|^m\leq \sum_{r=1}^{nk} (2n)^{2r} \exp\left(-\frac{rm}{k^2({2n}-1)}\right)=\sum_{r=1}^{nk} z^{r} \leq \frac{z}{1-z}.
\end{equation*}
Taking $m$ as in the hypothesis of the lemma then yields $z/(1-z)<\delta$, at which point we can appeal directly to Lemma~\ref{lem:fourier} to finish the proof.
\end{proof}

\lemparadepth*
\begin{proof}
Recall that we write the   random pairs of an $m$-step walk (in chronological order) as 
$e_1,\ldots,e_m$, and we assign to the 
 $i$\textsuperscript{th} step the level
\begin{equation}
 \ell_i=1+\max\{\ell_j:j<i,\ e_j\cap e_i\neq\{\}\},
\end{equation}
and put $D=\max_i\ell_i$. So, a gate having level at least $L$ means that we can find a chain $i_1<i_2<\cdots<i_L$, with $ e_{i_s}\cap e_{i_{s+1}}\neq\{\}$ for $1\leq s<L$. What is the probability of this? Well, take some arbitrary increasing list of indices $j_1,\ldots,j_L$. The probability that the gate at some index $j_s$ shares a Majorana in common with the gate at index $j_{s-1}$ is
\begin{equation*}
 p_{2n}=\frac{(2n-1) + (2n-1) - 1}{\binom {2n}2}
     =\frac{4n-3}{n(2n-1)},
\end{equation*}
and so the probability that the entire index tuple forms a chain is $p_{2n}^{L-1}$. There are $\binom mL$ choices of indices that could potentially form a chain, so that a union bound gives the overall probability to be bounded as 
\begin{equation*}
    \Pr[D\geq L]\leq \binom mL p_{2n}^{L-1} \leq p_{2n}^{-1}\left(\frac{em p_{2n}}{L}\right)^L,
\end{equation*}
where we have used $\binom mL\leq(em/L)^L$.  In particular, for $0<\eta<1$, let us make the choice
\begin{equation*}
 L=\left\lceil\max\left\{ 2em p_{2n},\ \log_2\frac{n}{\eta}\right\}\right\rceil.
\end{equation*}
Then
\begin{equation*}
    \Pr[D\geq L]\leq p_{2n}^{-1}\left(\frac{em p_{2n}}{2emp_{2n}}\right)^L = p_{2n}^{-1}2^{-L}\leq n2^{-L}\leq \eta,
\end{equation*}
where we have used that $n^{-1}\leq p_{2n}\to2n^{-1} $. 

\end{proof}

\lemmajodoubled*
\begin{proof}
Let us recall the channels of the various ensembles that are floating around at this point. We have $\Phi_{\mu_G}$, the channel corresponding to the true Haar distribution on matchgates, $M=\Phi_{\nu^{*m}}$, that of our $m$-step walk, and $T=\Phi_{\nu_{m,L}}$, the $m$-step walk channel conditioned on walks of  less than $L$ layers.  Let us also introduce $R$, the channel corresponding to walks conditioned on having  at least $L$ layers, so that 
\begin{equation}
  (1+\delta)\Phi_{\mu_G} \succeq  M = (1-r) T + rR \succeq (1-r) T\succeq (1-\eta) T,
\end{equation}
where $r$ is the exact probability of the walk being too long, and $\eta$ is from Lemma~\ref{lem:para_depth}, and we have  used that $\nu^{*m}$ is a $\delta$-approximate relative error design.
Then, using $\delta=\eta=\varepsilon/3$ and $0<\varepsilon\leq1$, we have
\begin{equation*}
T\preceq\frac{1+\varepsilon/3}{1-\varepsilon/3}\Phi_{\mu_G} \preceq(1+\varepsilon)\Phi_{\mu_G}, 
\end{equation*}
so that $\Delta=(1+\varepsilon)\Phi_{\mu_G}-T$ is completely positive. As compositions of completely positive maps are completely positive, we then have (using $\Phi_{\mu_G}^2=\Phi_{\mu_G} T = T\Phi_{\mu_G} = \Phi_{\mu_G}$) that
\begin{equation*}
    0 \preceq \Delta^2=T^2-(1-\varepsilon^2)\Phi_{\mu_G}
\end{equation*}
Combined with $0\preceq T\Delta=(1+\varepsilon)\Phi_{\mu_G}-T^2$, we conclude
\begin{equation*}
(1-\varepsilon)\Phi_{\mu_G}\preceq (1-\varepsilon^2)\Phi_{\mu_G} \preceq T^2 \preceq (1+\varepsilon)\Phi_{\mu_G}, 
\end{equation*}
i.e., $T^2$ is an $\varepsilon$-approximate relative error design. 
As $T^2$ is precisely the moment channel of the concatenation of two independently sampled accepted walks, i.e., that of $\mu_{m,L}$, we are done.
\end{proof}

\end{document}